\documentclass[12pt,reqno]{article}

\usepackage{amsmath,amssymb,amsbsy,amsthm}
\usepackage{url}
\usepackage{fullpage}
\usepackage{enumitem}

\usepackage{tikz}
\usepackage[linktocpage]{hyperref}
\hypersetup{
    colorlinks=true,        
    linkcolor=red,          
    citecolor=red,         
    filecolor=magenta,      
    urlcolor=cyan           
}

\addtocontents{toc}{\protect\setcounter{tocdepth}{1}}

\usepackage{eucal}

\newcommand{\N}{ \mathbb{N} }

\def\suchthat{\, : \,}
\def\andd{\, \wedge\, }
\def\theory{\textrm{FO}}
\DeclareMathOperator{\per}{per}
\DeclareMathOperator{\prefx}{\textsc{Pref}}
\DeclareMathOperator{\suffx}{\textsc{Suff}}
\DeclareMathOperator{\Fac}{Fac}
\DeclareMathOperator{\earliestfac}{\textsc{EarliestFac}}
\DeclareMathOperator{\period}{\textsc{Per}}
\DeclareMathOperator{\factoreq}{\textsc{FactorEq}}
\DeclareMathOperator{\match}{\textsc{Match}}
\DeclareMathOperator{\prim}{\textsc{Primitive}}

\theoremstyle{plain}
	\newtheorem{thm}{Theorem}
	\newtheorem{notn}[thm]{Notation}
		\numberwithin{thm}{section}
	\newtheorem{lemma}[thm]{Lemma}
	\newtheorem{prop}[thm]{Proposition}
	\newtheorem{cor}[thm]{Corollary}

	\newtheorem*{thm*}{Theorem}
	\newtheorem*{lemma*}{Lemma}
	\newtheorem*{prop*}{Proposition}
	\newtheorem*{cor*}{Corollary}
	\newtheorem*{conj*}{Conjecture}

\theoremstyle{definition}
	
	\newtheorem*{example*}{Example}
	\newtheorem{defn}[thm]{Definition}

\author{Jason Bell\\Department of
Pure Mathematics\\
University of Waterloo\\
Waterloo, ON  N2L 3G1 \\
Canada\\
\href{mailto:jpbell@uwaterloo.ca}{\tt jpbell@uwaterloo.ca}
\and
Jeffrey Shallit\\
School of Computer Science\\
University of Waterloo\\
Waterloo, ON  N2L 3G1 \\
Canada\\
\href{mailto:shallit@uwaterloo.ca}{\tt shallit@uwaterloo.ca}}

\title{Automatic Sequences of Rank Two}

\begin{document}

\maketitle

\begin{abstract}
Given a right-infinite word $\bf x$ over a finite alphabet $A$, the \emph{rank} of $\bf x$ is the size of the smallest set $S$ of words over $A$ such that $\bf x$ can be realized as an infinite concatenation of words in $S$. We show that the property of having rank two is decidable for the class of $k$-automatic words for each integer $k\ge 2$.  
\end{abstract}

\section{Introduction}

Let $k\ge 2$ be an integer. In this paper we study $k$-automatic sequences, which are those sequences (or infinite words) $(a_n)_{n \geq 0}$ over a finite alphabet, generated by a deterministic finite automaton with output (DFAO) taking, as input, the base-$k$ representation of $n$
and outputting $a_n$.  Many interesting classical examples of sequences, including the Thue-Morse sequence, the Rudin-Shapiro sequence, and the paper-folding sequence are in this class. For more information about this well-studied class of sequences, see, for example, \cite{Allouche&Shallit:2003}.  We mention that there is another well-known characterization of the $k$-automatic sequences, as the image, under a coding, of the fixed point of a $k$-uniform morphism 
\cite{Cobham:1972}.  Here a morphism is called $k$-uniform if the length of the image of every letter is $k$, and a coding is a $1$-uniform morphism.

Let $x$ be a finite nonempty word.  We define $x^\omega$
to be the one-sided infinite word
$xxx\cdots$.
We say that an infinite word $\bf z$
is {\it ultimately periodic\/} if there
exist finite words $y, x$, with $x$
nonempty, such that ${\bf z} = y x^\omega$.
Honkala proved \cite{Honkala:1986} that
the following problem is decidable:  given a DFAO representing a $k$-automatic sequence $\bf x$, is $\bf x$ ultimately periodic?   In fact, later results showed that this is actually {\it efficiently\/} decidable;  see Leroux \cite{Leroux:2005} and 
Marsault and Sakarovitch \cite{Marsault&Sakarovitch:2013}.   For other related work, 
see \cite{Linna:1984,Pansiot:1986,Harju&Linna:1986,Lando:1991,Bell&Charlier&Fraenkel&Rigo:2009,Charlier&Massuir&Rigo&Rowland:2020,Durand:2013}.

Let $L$ be a language.   We define
$L^\omega$ to be the set of infinite words
$$ \{ x_1 x_2 \cdots \suchthat x_i \in L\setminus \{ \epsilon \} \}.$$
If ${\bf x} \in L^\omega$ for some finite language $L$ consisting of $t$ nonempty words, then we say that ${\bf x}$ is of \emph{rank} $t$.   In particular, deciding the ultimate periodicity of $\bf x$ is the same as deciding if some suffix of $\bf x$ is of rank one.  More generally, the rank serves as an invariant that gives some measure of the complexity of a word. 

While the rank is a useful invariant in the study of words, it can nevertheless be difficult to determine its precise value. In fact, there is some relationship with undecidable ``tiling'' problems such as the Post correspondence problem, which asks whether, given two finite sets of words of the same size, $\{a_1,\ldots ,a_m\}$ and $\{b_1,\ldots ,b_m\}$, over a common alphabet, there exist $k\ge 1$ and $i_1,\ldots ,i_k\le m$ such that $a_{i_1}\cdots a_{i_k}=b_{i_1}\cdots b_{i_k}$.  In general it is the fact that words cannot always be tiled unambiguously that complicates decision procedures involving tilings and such problems typically become more complex as the number of tiles involved increases.

Our main result is to show that the property of being of rank two is decidable for automatic words.
\begin{thm}\label{thm:main}
Let $k\ge 2$ be a positive integer and let $\bf x$ be a $k$-automatic sequence.  Then there is an algorithm to decide whether $\bf x$ is of rank two.
\end{thm}
This algorithm is considerably more involved than the corresponding algorithm used for determining whether a word has rank one; moreover, we do not currently know how to extend our method to arbitrary suffixes of $\bf x$, nor to sequences of higher rank.

A key component in the procedure given in Theorem \ref{thm:main} is the following result, which shows that there is a striking dichotomy in the possible powers of words that can appear in a $k$-automatic sequence.
\begin{thm} 
Let $k\ge 2$ be a positive integer and let $\bf x$ be a $k$-automatic sequence.  Then there is a computable bound
$B$, depending only on $\bf x$, such that if $y^B$ occurs as a factor of $\bf x$ then
$y$ occurs with unbounded exponent in $\bf x$.  
\label{expbound}
\end{thm}

The outline of this paper is as follows.  In \S2 we give the notation that will be used throughout the paper.  In \S3, we provide the necessary background on repetitive words.  In \S4, we recall a key result in first-order logic and use it to deduce Theorem \ref{expbound} along with a key technical result that will be used in the proof of Theorem \ref{thm:main}.  Then in \S5, we prove the key combinatorial lemmas that will be used in giving the decision procedure in Theorem \ref{thm:main}.  Finally, in \S6, we give the proof of Theorem \ref{thm:main}.

\section{Notation and definitions}
Throughout this paper we will make use of the following notation and definitions.

If $w = xyz$ for words $w,x,y,z$ with $w$ and $z$ possibly infinite, we say that $y$ is a {\it factor\/} of $w$, $x$ is a {\it prefix\/} of $w$, and $z$ is a {\it suffix\/} of $w$.

Let ${\bf x} = a_0 a_1 a_2 \cdots$ be an infinite word.  By ${\bf x}[i..i+n-1]$ we mean the length-$n$ word $a_i \cdots a_{i+n-1}$.  By $\Fac({\bf x})$ we mean $\{ {\bf x}[i..i+n-1] \suchthat i, n \geq 0 \}$, the set of all finite factors of $\bf x$.

Given a finite word $w = a_1 a_2 \cdots a_n$
we say that $w$ is of \emph{period} $p$ if
$a_i = a_{i+p}$ for $1 \leq i \leq n-p$.  A word can have multiple periods; for example, the French word {\tt entente} has periods $3, 6, $ and $7$.
We refer to the smallest positive period as \textit{the\/} period of $w$, and denote it by $\per (w)$.   The \textit{exponent\/} of a finite word $w$ is defined to be $\exp(w) = |w|/\per(w)$.  

A finite nonempty word $w$ is {\it primitive} if it is a non-power, that is, if it cannot be written as $w = x^e$ for some $e \geq 2$.   
If $w$ appears in $\bf x$ to arbitrarily large powers, we say that $w$ is {\it of unbounded exponent\/} in $\bf x$.  If $\bf x$ has only finitely many primitive factors of unbounded
exponent, we say it is {\it discrete\/}.

We assume the reader has a basic background in formal languages and finite automata theory.   For the needed concepts, see, for example,
\cite{Hopcroft&Ullman:1979}.

\section{Repetitive words}

An infinite word $\bf x$ is called
{\it repetitive\/} if for all $n$ there exists a finite nonempty word $w$ such that
$w^n \in \Fac({\bf x})$.  It is called
{\it strongly repetitive\/} if there exists a finite nonempty word $w$ such that
$w^n \in \Fac({\bf x})$ for all $n$.

Suppose $h$ is a morphism and $a$ is a letter such that $h(a) = az$ for some $z$ for which $h^i(z) \not= \epsilon$ for all $i$.    Then we say that $h$ is
{\it prolongable\/} on $a$.  In this case
$h^\omega(a) := a \, z \, h(z) \, \ldots $
is an infinite word that is a fixed point of $h$, and we say that $h^\omega(a)$ is a {\it pure morphic\/} word.   If $\bf x$ is the image, under a coding, of a pure morphic word, we say that $\bf x$ is
{\it morphic}.

Several writers have investigated the repetitive and strongly repetitive properties of pure morphic words.   Ehrenfeucht and Rozenberg \cite{Ehrenfeucht&Rozenberg:1983} showed that a pure morphic word is repetitive if and only if it is strongly repetitive, and also showed that these conditions are decidable.   Mignosi and S\'e\'ebold \cite{Mignosi&Seebold:1993a} proved that for every morphic word $\bf x$ there exists a constant $M$ such that $w^M \in 
\Fac({\bf x})$ if and only if 
$w^n \in 
\Fac({\bf x})$ for all $n$; also see
\cite{Klouda&Starosta:2019}.
Kobayashi and Otto \cite{Kobayashi&Otto:2000} gave an efficient algorithm to test the repetitivity of a pure morphic word.   Klouda and Starosta \cite{Klouda&Starosta:2015} showed further that every pure morphic word $\bf x$ is discrete, while the authors \cite{Bell&Shallit:2021} proved that words of linear factor complexity are discrete.

Only the last of these results applies to the case that concerns us in this paper (where $\bf x$ is $k$-automatic) because automatic words need not be pure morphic.  For
example, it is not hard to show that the
Rudin-Shapiro sequence is not pure morphic.

Recently, in a thus-far unpublished manuscript, Klouda and Starosta
\cite{Klouda&Starosta:2021}  showed that
morphic words (and hence $k$-automatic words) are discrete.

\section{First-order logic}

Certain key parts of the decision procedure in Theorem \ref{thm:main} will rely on the following result, which is essentially a consequence of the results of Bruy{\`e}re et al.~\cite{Bruyere&Hansel&Michaux&Villemaire:1994}; also see 
\cite{Charlier&Rampersad&Shallit:2012}.
\begin{thm}
Let $\bf x$ be a $k$-automatic sequence, and let $\varphi$ be a first-order logic formula,  expressible in $\theory(\N, +, n \rightarrow {\bf x}[n])$.  Then
\begin{itemize}
    \item[(a)] If $\varphi$ has no unbound variables, then the truth of $\varphi$ is decidable. 
    \item[(b)] If $\varphi$ has unbound variables, we can computably determine a DFA that recognizes precisely the base-$k$ representations of those natural number values of the unbound variables that make $\varphi$ true.
\end{itemize}
\label{bruyere}
\end{thm}

As an application, for automatic sequences we get easy proofs of the decidability of the repetitive and strongly repetitive properties.

\begin{thm}
Let $\bf x$ be a $k$-automatic sequence, and let $z$ be a given nonempty word.   Then the following problems are decidable:
\begin{itemize}
    \item[(a)] Do arbitrarily large powers of $z$ appear in $\bf x$?
    \item[(b)] If the answer to (a) is no, what is the largest exponent $e$ such that $z^e$
    appears in $\bf x$?
\end{itemize}
\end{thm}

\begin{proof}
\leavevmode
\begin{itemize}
    \item[(a)]  Let $z = a_1 a_2 \cdots a_r$.  We can write a first-order statement asserting that $z$ is a factor of $\bf x$ as follows:
    \begin{equation}
    \exists i \ {\bf x}[i] = a_1 \andd 
    {\bf x}[i+1] = a_2 \andd \cdots \andd
    {\bf x}[i+r-1] = a_r ,
    \label{andy}
    \end{equation}
    and so it is decidable if this is the case.
    
    Let us now create a formula asserting that
    the length-$m$ factor beginning at position $j$ has period $r$, and ${\bf x}[i..i+r-1] = {\bf x}[j..j+r-1]$:
    \begin{align*}
        \factoreq(i,j,n) &= \forall t\leq n \ {\bf x}[i+t]={\bf x}[j+t] \\
    \match(i,j,m,r) &= \factoreq(i,j,r) \andd \factoreq(j,j+r,m).
    \end{align*}
    
    If $z$ does indeed appear in $\bf x$,
    we can identify
    some $i$ for which Eq.~\ref{andy} holds.  Then
    arbitrarily large powers of $z$ appear in $\bf x$ if and only if
    $$ \forall m \ \exists j \ \exists n>m 
   \match(i,j,m,r).$$
   Here $i$ and $r$ are constants and not free variables.
    
    \item[(b)]   If $z$ occurs in $\bf x$, but not with arbitrarily large powers, then we can determine the largest (fractional) power $z^e$ occurring in $\bf x$ as follows:  create the DFA corresponding to the logical formula
    $$\exists m \ (\exists j \ \match(i,j,m,r) \andd \neg\exists j' \ \match(i,j',m-1,r)) .$$
    Again, $i$ and $r$ are constants and not free variables.
    This DFA will accept the base-$k$ representation
    of exactly one $m$, and then 
    $e = m/r$.   
    \end{itemize}
\end{proof}

As an immediate consequence we can prove Theorem \ref{expbound}.
\begin{proof}[Proof of Theorem \ref{expbound}]
Let $M$ be a $k$-DFAO generating the sequence $\bf x$.  As is well-known \cite[Thm.~14]{Goc&Schaeffer&Shallit:2013}, for every automatic sequence $\bf x$
there is a computable constant $C$ such that
if $y$ appears as a factor of $\bf x$, it must appear starting
at a position that is $\leq C|y|$.

Now consider the following first-order formula $\varphi(i,n,p)$:
$$ \exists j \ \earliestfac(i,j,p) \ \wedge\ \period(j,n,p) ,$$
where
\begin{align*}
\period(i,n,p) &= \factoreq(i,i+p,n-p) \\
\earliestfac(i,j,n) &= \factoreq(i,j,n) \ \wedge\
\forall t \ \factoreq(t,j,n) \implies t\geq i.
\end{align*}
The formula $\varphi$ asserts that there exists some $j$ such that
\begin{itemize}
\item[(i)] ${\bf x}[i..i+p-1] = {\bf x}[j..j+p-1]$;
\item[(ii)] $i$ is the smallest index for which (i) holds;
\item[(iii)] 
${\bf x}[j..j+n-1]$ has period $p$.
\end{itemize}
Notice that this implies that ${\bf x}[j..j+n-1]$ has exponent
at least $n/p$.

From $M$ we can computably determine a DFA $M'$ accepting
those triples $(i,p,n)_k$ in parallel making $\varphi(i,n,p)$ true.
Suppose $M'$ has $r$ states.
We now claim that $M'$ accepts $(i,p,n)_k$ with $n/p > k^r C$
if and only if $\bf x$ contains arbitrarily large powers of ${\bf x}[i..i+n-1]$.

One direction is trivial.   For the other direction, consider
the base-$k$ representation of the triple $(i,p,n)$.
From the discussion above we know that $i \leq Cp$.
Since $n > k^r Cp$ we have $n > k^r i$, and hence the
base-$k$ representation of $(i,p,n)$ starts with at least
$r$ $0$'s in the components corresponding to $i$ and $p$
and a nonzero digit in the $n$ component.   We may now
apply the pumping lemma to $z = (i,p,n)_k$ to get longer and
longer strings with the same value of $i$ and $p$, but
arbitrarily large $n$.   From the definition of $M'$
this means that there is an infinite sequence of
increasing $n$ for which there exists a $j$
with ${\bf x}[j..j+n-1]$ of exponent at least $n/p$.
We may now take $B = k^r C$ to prove the result.
\end{proof}

We can also prove that $k$-automatic words are discrete, and even something more general.  The following result appears in earlier work of the authors \cite[Theorem 1.2]{Bell&Shallit:2021}, but we give a simpler and more self-contained proof here. We note, however, that the following proof does not recover the upper bound on the number of primitive factors that can occur with unbounded exponent (up to cyclic equivalence), which was given in \cite{Bell&Shallit:2021}.
\begin{thm}
Let $\bf x$ be an infinite word with linear subword complexity.  Then $\bf x$ is discrete.
\end{thm}

\begin{proof}
Suppose that the subword complexity of $\bf x$ is bounded by $cn$ for some constant $c$ and $n \geq 1$.   Further suppose, contrary to what we want to prove, that $\bf x$ has infinitely many primitive factors $x_1, x_2, \ldots, x_{2c+1}$ for which arbitrarily large powers appear in $\bf x$.

Choose $2c+1$ of them that are strictly increasing in size,
$1 \leq |x_1| < |x_2| < \cdots < |x_{2c+1}|$.  Replace each $x_i$
with $y_i$, an appropriate power of the $x_i$ so that all the $y_i$ are
the same length $d >c$.   (The $y_i$ are no longer primitive, but it doesn't matter.)

Pick $n = 3d+1$.   For each word $y_i$ find an occurrence of a power of $y_i$ in $\bf x$ to the 3rd power.
Without loss of generality we may
assume (by replacing a $y_i$ by a cyclic shift of it, if necessary)
that the occurrence of this $y_i^3$ in $\bf x$ is followed by a letter $a$ different from the first letter of $y_i$.

For each $i$, we take a suffix ${\bf x}_i$ of $\bf x$ that has $y_i^3$ as a prefix. Now consider the list $L_i$ of the $2|y_i|+1$ length-$n$ factors of $\bf x$ that start at position $p$ of ${\bf x}_i'$ for $p=0,\ldots ,2|y_i|$.  
that starts $y_i$.   For each $y_i$ there are $2|y_i|+1 = 2d+1$ such words.   Suppose two of these shifted words agree.
Then either $y$ matches a shift of $y$, forcing a mismatch at the letter $a$, or $y$ doesn't match a shift of $y$, forcing a mismatch within the first $|y|$ symbols.

Similarly, if we compare a word of $L_i$ and a word of $L_j$ for $i \not= j$, then the fact that they agree on a prefix
of at least one repetition of the same length-$d$ word would imply $y_i$ is a cyclic shift of $y_j$, impossible since the original $x_i$ were
distinct primitive words.

So we have constructed at least $2d(2c+1) = 4cd + 2d > c(3d+1) $ different length-$n$ subwords of $\bf x$, a contradiction.
\end{proof}

\begin{cor}
If $\bf x$ is $k$-automatic, we can explicitly determine the (finitely many) primitive factors $w$ such that $w$ is of unbounded exponent in $\bf x$.
\label{simple}
\end{cor}

\begin{proof}
We can easily write down a first-order formula asserting that
$w = {\bf x}[i..i+p-1]$ is primitive, that it is the first occurrence of this factor in $\bf x$, and that unboundedly large powers of $w$ appear in $\bf x$, as
follows:
$$ \forall m \ \exists j, n \ (n>m) \ \wedge\  \prim(i,p) \ \wedge\ \earliestfac(i,j,p) \ \wedge\ 
 \period(j,n,p), $$
 where
\begin{align*}
\prim(i,n) &:= \neg(\exists j\ (j>0) \ \wedge\ (j<n) \ \wedge\ \factoreq(i,i+j,n-j) \ \wedge\  \\
& \factoreq(i,(i+n)-j, j)). \\
\end{align*}
So by Theorem~\ref{bruyere} we can compute an automaton recognizing the finitely many pairs $i, n$.
\end{proof}

Finally, we prove two technical results, which together will play a key role in the proof of Theorem \ref{thm:main}.
\begin{prop}
\label{prop:setup}
Given a $k$-automatic sequence $\bf x$, and a nonempty factor $u = {\bf x}[i..i+d-1]$, and fixed
natural numbers $L, N$, one can decide if there exist a word $v$ of length $\geq N$, such that $u$ is neither a prefix or suffix of $v$, and natural numbers $p_1, p_2, \ldots, p_L$ such that $v u^{p_1} v u^{p_2} v \cdots v u^{p_L}$
is a prefix of $\bf x$.
\end{prop}

\begin{proof}
We first construct a first-order logical formula 
for prefix and suffix by
\begin{align*}
\prefx(i,j,x,y) &= j \leq y \andd \factoreq(i,x,j) \\
\suffx(i,j,x,y) &= j \leq y \andd \factoreq(i,y-j,j).
\end{align*}
The former asserts that ${\bf x}[i..i+j-1]$ is
a prefix of ${\bf x}[x..x+y-1]$, and the latter
asserts that ${\bf x}[i..i+j-1]$ is a suffix
of ${\bf x}[x..x+y-1]$.
Then a logical formula asserting
the existence of $v$ and the $p_i$ is as follows:
\begin{multline*}
r \geq N \andd 
\neg\prefx(i,d,0,r) \andd 
\neg\suffx(i,d,0,r) \andd 
\exists r, p_1, p_2, \ldots, p_L \\
\factoreq(0,r+p_1 d,r) \andd
\factoreq(0,2r+(p_1+p_2)d,r) \andd \cdots \andd \\
\factoreq(0,(L-1)r,(p_1+p_2+\cdots+p_{L-1})d) 
\andd \\
\period(r,p_1d,d) \andd \period(2r+p_1 d,p_2d,d) \andd \cdots \andd \period(Lr+(p_1+p_2+\cdots+p_{L-1})d, p_{Ld}, d) .
\end{multline*}
Notice that $d$ and $L$ are fixed constants, 
and that we are defining $v = {\bf x}[0..r-1]$
and $u = {\bf x}[i..i+d-1]$.  This statement is
true if and only if the desired $v$ exists.
If it is true, then we can easily find the
smallest $r$ for which it is true, just
as we did above in the proof of Corollary~\ref{simple}.
\end{proof}

\begin{prop}
\label{prop:setup2}
Let $k,m\ge 2$ be integers. Given a $k$-automatic sequence $\bf x$, and a binary word $i_0\cdots i_{m-1}\in \{0,1\}^m$, we can decide whether there exist nonempty factors $u_0$ and $v_1$ of $\bf x$ such that $u_0$ is neither a prefix nor suffix of $u_1$; $u_1$ is neither a prefix nor suffix of $u_0$; and such that 
$u_{i_0}\cdots u_{i_{m-1}}$ is a prefix of $\bf x$.
\end{prop}
\begin{proof} We can construct a first-order formula encoding these assertions.   The idea is that $u_0 = {\bf x}[i..i+r-1]$ and $u_1 = {\bf x}[j..j+s-1]$ for some $i,j,r,s$.   If $u_{i_0}\cdots u_{i_{m-1}}$ is a prefix of $\bf x$, then there exist starting positions $p_0=0, p_1, \ldots, p_{m-1}$ and lengths
$q_0, q_1, \ldots, q_{m-1}$ corresponding to each of the occurrences of the $u_{i_t}$.
We then assert that the starting positions obey the rule that $p_{t+1} = p_t + q_t$ for $0 \leq t < m-1$, and that each occurrence $u_{i_t}$ match ${\bf x}[i..i+r-1]$ or
${\bf x}[j..i+s-1]$, according to whether $i_t = 0$ or $i_t = 1$, respectively.   This gives us the following first-order formula:
\begin{align*}
&\exists i,j,r,s, p_0, p_1, \ldots, p_{m-1}, q_0, q_1, \ldots, q_{m-1}  \ 
r>0 \ \wedge\ s>0 \ \wedge \ \\
& \neg \prefx(i,r,j,s) \ \wedge\ 
\neg \suffx(i,r,j,s) \ \wedge\ 
\neg \prefx(j,s,i,r) \ \wedge\ 
\neg \suffx(j,s,i,r) \ \wedge \ \\
&\forall t<m \ ((i_t=0) \implies
q_t = r \ \wedge\ \factoreq(i, p_t, q_t)) \ \wedge\ \\
& ((i_t=1) \implies q_t = s \ \wedge\ \factoreq(j,p_t, q_t))) \ \wedge \\
& p_0 = 0 \ \wedge\ \forall t<m-1 \ 
p_{t+1} = p_t + q_t .
\end{align*}
\end{proof}
\section{Combinatorial lemmas}
In this section, we prove result about semigroup equations, which will again play a key role in our decision procedure for testing rank two.

Given a finite alphabet $\Sigma$, we write $a\le b$ for $a,b\in\Sigma^*$ if $a$ is a prefix of $b$.
We note that if $a$ and $b$ have the property that $a\not\le b$ and $b\not\le a$ then $a$ and $b$ generate a free semigroup and every word in $\Sigma^*$ has a unique largest prefix in $\{a,b\}^*$.
\begin{lemma} Let $\Sigma$ be a finite alphabet and let $r$ and $s$ be elements in $\Sigma^*$ that do not commute.  Then there exist $a,b\in\Sigma^*$ such that $a\not\le b$ and $b\not\le a$ and such that 
$r,s\in \{a,b\}^*$.
\label{lem:rsab}
\end{lemma}
\begin{proof} Suppose towards a contradiction that the conclusion to the statement of the lemma does not hold.  Then among all counterexamples $(r,s)\in (\Sigma^*)^2$, we pick $(r,s)$ with $|r|+|s|$ minimal.  
Then either $r\le s$ or $s\le r$, or else we can take $a=r$ and $b=s$ and we obtain a contradiction. Thus we may assume without loss of generality that $r\le s$ and so we write $s=ry$.  Then $r,s\in \{r,y\}^*$.  Then either $|r|+|y|<|r|+|s|$ or $r=\epsilon$.  But we cannot have $r=\epsilon$ since $r$ and $s$ do not commute. Since $s$ and $r$ do not commute, $r$ and $y$ do not commute and by minimality of $|r|+|s|$, we have $\{r,y\}^*\subseteq \{a,b\}^*$ with $a\not\le b$ and $b\not\le a$.  But this is a contradiction, since $r,s \in \{r,y\}^*\subseteq \{a,b\}^*$.  The result follows.
\end{proof}
\begin{lemma}  
Let $\Sigma$ be a finite alphabet, let $u,v\in \Sigma^*$ be non-trivial words in $\Sigma^*$ such that there do not exist $a,b\in \Sigma^*$ with $|a|+|b|<|u|+|v|$ and with $u,v\in \{a,b\}^*$, and let $\sigma: \{x,y\}^*\to \{u,v\}^*$ denote the homomorphism from the free monoid on the set $\{x,y\}$ to the monoid generated by $u$ and $v$ given by $\sigma(x)=u$ and $\sigma(y)=v$.  Suppose that there exist words $d,d'\in \Sigma^*$, and words $w,w'\in \{x,y\}^*$ such that the following hold:
\begin{enumerate}
\item $d\sigma(w)= \sigma(w')d'$;
\item $w,w'$ both have $x$ as a prefix and at least one occurrence of $y$;
\item $d$ does not contain $u$ as a suffix and $|d|<\max(|u|,|v|)$.
\end{enumerate} 
Then $d=\epsilon$.
\label{lem:depsilon}
\end{lemma}
\begin{proof} 
Let $z=d\sigma(w)= \sigma(w')d'$. By considering the prefix of $z$ of length $|du|$, we see $du = u c$ for some $c$ and hence by the Lyndon-Sch\"utzenberger theorem \cite{Lyndon&Schutzenberger:1962}
there exist words $r,s$ such that $u= (rs)^{\alpha} r$, $d=rs$, $c=sr$.  Notice that if $r$ and $s$ commute, then again by the Lyndon-Sch\"utzenberger theorem, they are powers of some word $t$, which then gives $u=t^i$ and $d=t^j$.  But now $\{u,v\}\in \{t,v\}^*$ and so $|t|+|v|=|u|+|v|$ by hypothesis and so $u=t$ and $i=1$.  But then since $u$ is not a suffix of $d$, we must have $d=\epsilon$.  Hence we may assume that $r$ and $s$ do not commute. 
Now by Lemma \ref{lem:rsab}, there are words $a$ and $b$ such that $a\not \le b$ and $b\not\le a$ such that the monoid generated by $r$ and $s$ is contained in the free monoid generated by $a$ and $b$.  Then there exist unique words $\alpha,\beta,\gamma\in \{a,b\}^*$ such that
$d=\alpha$, $u=\beta$ and $v=\gamma v'$ where $v'$ does not have $a$ or $b$ as a prefix.  If $v'=\epsilon$ then $u,v\in \{a,b\}^*$ and by assumption, we must have $u=a$ and $v=b$ after relabelling.  But now $d=\alpha$ is such that $|d|<\max(|a|,|b|)$ and since $\alpha\in \{a,b\}^*$, we see that $d$ is either a power of $a$ or a power of $b$.  But $d$ cannot be $a^i$ with $i\ge 1$ since $d$ does not have $u=a$ as a suffix and hence $d$ must be $b^j$.  But now the equation $du=uc$ gives $b^j a = ac$, which is impossible since $b\not\le a$ and $a\not\le b$.  Thus $v'\neq \epsilon$.

Then by assumption $w$ has at least one occurrence of $y$ and hence there is some $i\ge 1$ such that $x^iy$ is a prefix of $w$.  Then 
$cu^i v$ is a prefix of $z$ and so the longest prefix of $z$ in $\{a,b\}^*$ is $\alpha\beta^i \gamma$.  On the other hand, there is some $j\ge 1$ such that 
$w'$ has $x^jy$ as a prefix and so the longest prefix of $z$ in $\{a,b\}^*$ is $\beta^j\gamma$.  Then we must have
$\alpha\beta^i\gamma =\beta^j\gamma$.  It follows that $j\ge i$ since $\beta$ is non-trivial.  Cancelling $\beta^i\gamma$ on the right gives $\alpha = \beta^{j-i}$, which gives $d=\tau(\alpha)=\tau(\beta)^{j-i} = u^{j-i}$.  Since $d$ does not have $u$ as a suffix, we then see that $j=i$ and so $d=\epsilon$, as required.
\end{proof} 
\begin{notn}\label{notn:standing}
For the remainder of this section, we adopt the following notation and assumptions:
\begin{itemize}
    \item $u$ and $v$ are words with the property there do not exist $a,b\in \Sigma^*$ with $|a|+|b|<|u|+|v|$ and with $u,v\in \{a,b\}^*$; 
    \item we let $\sigma: \{x,y\}^*\to \{u,v\}^*$ denote the homomorphism from the free monoid on the set $\{x,y\}$ to the monoid generated by $u$ and $v$ given by $\sigma(x)=u$ and $\sigma(y)=v$. 
\end{itemize}
\end{notn}
\begin{prop} Let $w\in \{x,y\}^*$ be a word that contains at least five occurrences of $xy$ and let $z$ be a word in $\Sigma^*$ with $|z|=\max(|u|,|v|)$ and such that $z$ does not contains $u$ or $v$ as a prefix.  Then $\sigma(w) z$ is not a factor of a word in $\{u,v\}^{\omega}$. 
\label{prop:comb}
\end{prop}
\begin{proof} Since $w$ has at least five occurrences of $xy$, we can write $w$ in the form $w_0xyw_1$ with $w_0$ and $w_1$ both having at least two occurrences of $xy$.

By assumption there is some word $w'=\xi_1\cdots \xi_s\in \{x,y\}^*$ such that $\sigma(w) z$ is a factor of $\sigma(w')$, and so we can write
$\sigma(w') = a \sigma(w_0) uv \sigma(w_1)zb$ for some words $a$ and $b$.  Then there must be some largest $i$ such that $w'':=\sigma(\xi_1)\cdots \sigma(\xi_i)$ is a prefix of $a\sigma(w_0)u$.
We now argue via cases.
\bigskip

\noindent{\bf Case I:} $\xi_{i+1}=y$.
\vskip 2mm
In this case, $w'' v$ is not a prefix of $a\sigma(w_0) u$, but it is a prefix of $a\sigma(w_0) uv$.  In particular, we can factor $u=u_0d$ such that $\sigma(w'')=a\sigma(w_0) u_0$ and we have
$$dv\sigma(w_1) z b = v \sigma(\xi_{i+2})\cdots \sigma(\xi_s).$$  Notice $|d|\le |u|\le \max(|u|,|v|)$, and we cannot have $|d|=|u|$, or else $u$ and $v$ would share a prefix, which cannot happen by hypothesis.  Thus $|d|<\max(|u|,|v|)$ and $d$ cannot have $v$ as a suffix, since otherwise $v$ would be a suffix of $u$ and so we could write $u=av$ and then $u,v\in \{c,v\}^*$ with $|c|+|v|<|u|+|v|$, which is a contradiction.  Also, we cannot have $d=\epsilon$, since $v\not\le u$ and $u\not\le v$, and this would imply that $w_1$ is a prefix of $\xi_{i+2}\cdots \xi_s$, and so $zb$ would have to be a prefix of some word of the form
$\sigma(\xi_j)\cdots \sigma(\xi_s)$.  But this would then say that either $u$ or $v$ is a prefix of $z$, since $|z|=\max(|u|,|v|)$ and this is a contradiction.

Now there is some smallest $j$ such that $w'':=v\sigma(\xi_{i+2})\cdots \sigma(\xi_j)$ contains $dv$ as a prefix.  Then if $\xi_{i+2}=\cdots = \xi_j=y$, then $v$ is a factor of $\sigma(\xi_{j-1})\sigma(\xi_j)=vv$ and since $v$ is primitive, this forces 
$dv=v\sigma(\xi_{i+2})\cdots \sigma(\xi_j)$ (see, e.g., \cite[p.~336]{Choffrut&Karhumaki:1997}).  In particular, this means $d$ would be a power of $v$, contradicting the fact that it cannot contain $v$ as a suffix.  

By assumption $w_1$ has at least one occurrence of $xy$ and so we may write $w_1=y^q x w_2$ for some $q\ge 0$ and $w_2$ containing at least one copy of $x$ and $y$.  Then there is some smallest $k$ such that
$v\sigma(\xi_{i+2})\cdots \sigma(\xi_k)$ contains $dv^{q+1}u$ as a prefix, and since $w_2$ contains a copy of at least one $x$ and $y$, we see that 
$v\sigma(\xi_{i+2})\cdots \sigma(\xi_k)$ is a prefix of $dv^{q+1}u \sigma(w_2)$.  Hence we can write
$$dv^{q+1}u \sigma(w_2)= v\sigma(\xi_{i+2})\cdots \sigma(\xi_k)d'.$$  Since $j\le k$, we see that $\xi_{\ell}=x$ for some $\ell\in \{i+2,\ldots,k\}$.
Then Lemma \ref{lem:depsilon} now gives that $d=\epsilon$, which we have ruled out.  Thus we have completed the proof in this case.
\bigskip

\noindent{\bf Case II:} $\xi_{i+1}=x$. \vskip 2mm
 In this case there is some word $d$ with $|d|\le |u|$ such that $\sigma(\xi_1)\cdots \sigma(\xi_i)ud=a\sigma(w_0) u$.  
 Arguing as in Case I, we can show that some letter in $w_0$ is equal to $y$ and that some $\xi_j$ with $j\le i$ is equal to $y$.  Then applying Lemma \ref{lem:depsilon} in the opposite monoid now gives $d=\epsilon$.  But we obtain a contradiction in this case as in Case I.  
 \bigskip
 This completes the proof.
 \end{proof}
\begin{defn}
 Let $p$ be a positive integer. We say that a word in $\{u,v\}^{\omega}$ is $p$-\emph{syndetic} if it has no factors of the form $u v^j u$ or $v u^j v$ with $j\ge p$. 
 \end{defn}
 
 \begin{lemma} Let $k\ge 2$ be a positive integer, let $\bf x$ be a $k$-automatic word, and let $u,v \in \Sigma^*$ be such that the assumptions of Notation \ref{notn:standing} hold. Suppose further that $\Fac(u^{\omega})\not\subseteq \Fac({\bf x})$ and $\Fac(v^{\omega})\not\subseteq \Fac({\bf x})$.  Then there is a computable integer $D=D({\bf x})>0$ such that if $\bf x$ has a prefix in $\{u,v\}^D$ then $\bf x$ is in $\{u,v\}^{\omega}$.
 \label{lem:D}
 \end{lemma}
 \begin{proof}
By Theorem \ref{expbound}, there is some computable $p=p({\bf x})$ such that $u^p$ and $v^p$ are not factors of $\bf x$.  By a result of Cobham \cite{Cobham:1972}, there is a computable number $\kappa=\kappa({\bf x})$ such that every factor of $\bf x$ of length $L$ occurs in the prefix of $\bf x$ of length $\kappa L$.  
We now take $D=10 p^2\kappa+p+1$ and suppose that $\bf x$ has a prefix in $\{u,v\}^D$ but that it is not in $\{u,v\}^{\omega}$.  Then there is some largest $d\ge D$ such that 
$\bf x$ has a prefix in $\{u,v\}^d$.  Then $\bf x$ has a prefix of the form $az$ with $a\in \{u,v\}^d$ and $|z|=\max(|u|,|v|)$ and not having $u$ or $v$ as a prefix.  Now $a$ is necessarily $p$-syndetic since $u^p$ and $v^p$ are not factors of $w$.  It follows that every factor of $w$ in $\{u,v\}^{2p-1}$ must have at least one occurrence of $uv$, and so we can write $a=bc$ where $b\in \{u,v\}^{d-5(2p-1)}$ and $c\in \{u,v\}^{10p-5}$.  Then $c$ has at least five occurrences of $uv$ and so $cz$ is not a factor of an element of $\{u,v\}^{\omega}$ by Proposition \ref{prop:comb}.  In particular, $cz$ is not a factor of $a$.  Since $cz$ is a factor of $w$ of length at most 
$$(10p-5)\max(|u|,|v|)+|z| \le 10 p \max(|u|,|v|),$$ we see that $cz$ must occur in a prefix of $w$ of length $10p \kappa  \max(|u|,|v|)$.  Thus since $cz$ is not a factor of $a$, we must have that 
$$|a| < 10p \kappa  \max(|u|,|v|).$$  But $a\in \{u,v\}^d$ has no occurrences of $u^p$ or $v^p$ and hence it must have at least $\lfloor d/p\rfloor $ copies of $u$ and at least $\lfloor d/p \rfloor$ copies of $v$ and so 
$$|a| \ge (d/p-1) \max(|u|,|v|) \ge (D-p)\max(|u|,|v|)/p.$$
But this now gives $D < 10 p^2 \kappa+p+1$, a contradiction.
\end{proof}
\section{Proof of Theorem \ref{thm:main}}

We now give the decision procedure that makes up the content of Theorem \ref{thm:main}; namely, we show how to decide whether there exist finite words $u$ and $v$ such that ${\bf x}\in \{u,v\}^{\omega}$, when $\bf x$ is a $k$-automatic sequence. The procedure is divided into two cases, which depend upon whether one of the words $u$ or $v$ has arbitrarily large powers occurring as factors of $\bf x$.  The former case is dealt with via using the following lemma and proposition.

\begin{lemma} \label{lem:L}
Let $k\ge 2$ be a positive integer, let $\bf x$ be a $k$-automatic sequence, and let $u$ be a nontrivial factor of $\bf x$.  Then there is a computable number $L=L({\bf x})$ such that if there exists a prefix $v$ of $\bf x$ with $|v|\ge |u|$ such that
\begin{enumerate}
    \item 
$u$ is not a prefix nor suffix of $v$,
\item $v$ does not occur in $\bf x$ with unbounded exponent,
\item $v$ is not a factor of $u^{\omega}$,
\item there exist $p_1,\ldots ,p_L\ge 0$ with the property that 
$vu^{p_1} v u^{p_2}v\cdots v u^{p_L}$ is a prefix of $\bf x$,
\end{enumerate}
then ${\bf x} \in \{u,v\}^{\omega}$. 
\end{lemma}
\begin{proof}
We recall that by a result of Cobham \cite{Cobham:1972} there is a computable number $\kappa=\kappa({\bf x})$ such that every factor of $\bf x$ of length $N$ occurs in the prefix of $\bf x$ of length $\kappa N$. Theorem \ref{expbound} gives that there is a computable number $p=p({\bf x})$ such that every factor of $\bf x$ has the property that it either occurs in $\bf x$ with exponent at most $p$ or it occurs with unbounded exponent. 

We take $L=(15p+4)\kappa$ and we claim that if there exists a prefix $v$ of $\bf x$ with $|v|\ge |u|$ such that $u$ is not a prefix nor suffix of $v$ and such that there are $p_1,\ldots ,p_L\ge 0$ with the property that 
$vu^{p_1} v u^{p_2}v\cdots v u^{p_L}$ is a prefix of $\bf x$ then ${\bf x}\in \{u,v\}^{\omega}$.

To see this, suppose towards a contradiction that ${\bf x}\not\in \{u,v\}^{\omega}$.  Then after possibly enlarging $L$, we may assume that $vu^{p_1} v u^{p_2}v\cdots v u^{p_L}$ is a prefix of $\bf x$ but neither $vu^{p_1} v u^{p_2}v\cdots v u^{p_L+1}$ nor $vu^{p_1} v u^{p_2}v\cdots v u^{p_L}v$ are prefixes of $\bf x$.  Then there is some word $z$ with $|z|=|v|$ such that $vu^{p_1} v u^{p_2}v\cdots v u^{p_L}z$ is a prefix of $\bf x$.

To complete the proof, we now look at cases.
\bigskip

\noindent {\bf Case I.} For each $i\in \{L-5p, \ldots ,L\}$, we have $|p_i|\cdot |u| \le |2v|$.  
\bigskip

In this case, 
$$vu^{p_{L-5p}} v \cdots v u^{p_L}$$ has at least five occurrences of $vu$, since $v^p$ is not a factor of $u$ and so by Proposition \ref{prop:comb},
$vu^{p_{L-5p}} v \cdots v u^{p_L}z$ is not a factor of a word in $\{u,v\}^{\omega}$.   Notice that the length of 
$$y:=vu^{p_{L-5p}} v \cdots v u^{p_L}z$$ is at most 
$|v| (15p+4)$, since each $u^{p_j}$ factor has length at most $2|v|$ and $|z|=|v|$.  Then by Cobham's result \cite{Cobham:1972} this word $y$ must occur in a prefix of $\bf x$ of length
$\kappa |v| (15p+4)$.  But since $vu^{p_1} v u^{p_2}v\cdots v u^{p_L}$ has length at least $L|v| > \kappa |v|(15p+4)$, we see that this cannot be the case.  
\bigskip

\noindent{\bf Case II.} There is some $i\in  \{L-5p, \ldots ,L\}$ such that $|p_i|\cdot |u| > |2v|$. 

\bigskip 
In this case, there is some maximal $i$ in this interval with this property, and we let $j$ denote this index.  Then there is some $q\le p_j$ such that $q|u| > 2|v| \ge (q-1)|v|$, and so $|u|^q \le 2|v|+ |u| < 3|v|$.   
Then we consider the suffix $y':=u^q v u^{p_{j+1}} v \cdots v u^{p_L} z$ of $vu^{p_1} v u^{p_2}v\cdots v u^{p_L}z$.  Notice that 
$|y'| \le 3|v|(1+L-j)+|v| \le 3|v|(5p+1)$.  Since $vu^{p_1} v u^{p_2}v\cdots v u^{p_L}$ has length at least $L|v| > \kappa |v|(15p+3)$, we see by Cobham's theorem \cite{Cobham:1972} that 
$y'$ must be a factor of $vu^{p_1} v u^{p_2}v\cdots v u^{p_L+1}$.  That is, there are words $a$ and $b$ such that 
$ay'b =\xi_1\cdots \xi_t$ with each $\xi_i\in \{u,v\}$.  Since $u\not\le v$ and $v\not\le u$, we may assume that neither $a$ nor $b$ is trivial, neither $u$ nor $v$ is a prefix of $a$, and neither $u$ nor $v$ is a suffix of $b$ and we may assume that $a$ is shorter than the length of $\xi_1$ and that $b$ is shorter than the length of $\xi_t$.  Then since $|u|^q > 2|v| \ge 2|u|$, we must have that $\xi_2$ is a factor of $u^q$.  By assumption $v$ is not a factor of $u^{\omega}$ and so $\xi_2$ must be $u$.  But now $\xi_2=u$ and so $u$ is in fact a factor of $u^2$.  Since $u$ is primitive, we know (see, e.g., \cite[p.~336]{Choffrut&Karhumaki:1997}) that if $a'ub'= u^2$, then either $a'=\epsilon$ or $b'=\epsilon$, and so we see that $\xi_1= au^i$ for some $i\ge 0$.  But now if $\xi_1=v$, then $u$ is a prefix of $v$, which is not allowed, and if $\xi_1=u$ either $a=\epsilon$ or $a=u$, neither of which is allowed. 
\bigskip 

This completes the proof.
\end{proof}

\begin{prop}\label{prop:uprim}
Let $\bf x$ be a $k$-automatic word and let $u$ be a primitive factor of $\bf x$ with the property that $\Fac(u^{\omega})\subseteq \Fac({\bf x})$. Then there is a decision procedure which decides whether there is a word $v$ such that ${\bf x}\in \{u,v\}^{\omega}$.
\end{prop}
\begin{proof}
We recall that by Corollary \ref{simple} there is a finite set $\{w_1,\ldots ,w_r\}$ of primitive factors of $\bf x$ that occur with unbounded exponents, which we can explicitly determine.  Then by assumption $u\in \{w_1,\ldots ,w_r\}$.  For two fixed words $a$ and $b$, it is straightforward to decide whether ${\bf x}\in \{a,b\}^{\omega}$, and so we may decide in the case that $v=w_j$ for some $j$ and when $|v|\le |u|$.  

Hence it suffices to deal with the case when $v\not\in \{w_1,\ldots ,w_r\}$ and $|v|>|u|$.  Moreover, by removing a prefix of $\bf x$ of the form $u^i$, we may assume without loss of generality that $u$ is not a prefix of $\bf x$; and we may assume without loss of generality that $u$ is neither a prefix nor a suffix of $v$ and that $v$ is primitive.  

Since $v$ is a prefix of $\bf x$ and since there is a unique longest prefix of $\bf x$ that is in ${\rm Fac}(u^{\omega})$, we can decide whether there exists $v\in {\rm Fac}(u^{\omega})$ such that ${\bf x}\in \{u,v\}^{\omega}$.  

Thus we may assume, in addition to the other assumptions given, that $v$ is not a factor of $u^{\omega}$.  It follows from Lemma \ref{lem:L} that there is a computable number $L=L({\bf x})$ such that if there exist $p_1,\ldots ,p_L\ge 0$ with the property that $v u^{p_1} v u^{p_2} \cdots v u^{p_L}$ is a prefix of $\bf x$, then ${\bf x}\in \{u,v\}^{\omega}$.  By Proposition \ref{prop:setup}, it is decidable whether $\bf x$ has a prefix of the form $v u^{p_1} v u^{p_2} \cdots v u^{p_L}$ for some $v$ having the desired constraints and some choice of $p_1,\ldots ,p_L$, and so we are done. 
\end{proof}

\begin{proof}[Proof of Theorem \ref{thm:main}]
We give the steps in the algorithm, which determines whether the rank of $\bf x$ is two.  Since the property of being periodic is decidable, we may assume that the rank of $\bf x$ is at least two. We note that if $\bf x$ is of rank two, then there exist words $u$ and $v$ such that ${\bf x}\in \{u,v\}^{\omega}$; then by picking such $(u,v)$ with $|u|+|v|$ minimal, we may assume without loss of generality that the assumptions from Notation \ref{notn:standing} hold.
\begin{enumerate}
\setlength{\itemindent}{.3in}
\item[Step 1.] Using Theorem~\ref{expbound}, compute $p=p(w)$ such that for every $u$ with the property that $u^p$ is a factor of $\bf x$ we have $\Fac(u^{\omega})\subseteq \Fac({\bf x})$.  

\item[Step 2.] By Corollary~\ref{simple}, there is a finite computable set of primitive words $\{w_1,\ldots, w_r\}$ such that $\Fac(w_i^{\omega}) \subseteq \Fac({\bf x})$ for $i=1,\ldots , r$.

\item[Step 3.] Use the decision procedure from Proposition \ref{prop:uprim} to decide whether there exists a word $u$ such that ${\bf x}\in \{w_i, u\}^{\omega}$ for some $u$ and some $i\in \{1,\ldots ,r\}$.  If such a $u$ exists, the algorithm halts and returns that $\bf x$ has rank two; if no such $u$ exists, we go to the next step. 
\item[Step 4.] It now suffices to decide whether there exist words $u,v$ such that $\bf x$ is in $\{u,v\}^{\omega}$ such that the assumptions of Notation \ref{notn:standing} apply to $u$ and $v$; moreover, in addition we may assume that $u,v\not\in \{w_1,\ldots ,w_r\}$, where $w_1,\ldots ,w_r$ are as in Step 2. Thus $u^p$ and $v^p$ are not factors of $\bf x$.  Then compute the integer $D$ given in the statement of Lemma \ref{lem:D}.  
\item[Step 5.] For each of the $2^D$ binary words $y$ of length $D$, use Proposition \ref{prop:setup2} to determine whether there exist $u$ and $v$ such that $y(u,v)$ is a prefix of $\bf x$; if there is some binary word for which this holds then $\bf x$ has rank two by Lemma \ref{lem:D} and we stop; if this does not hold for these words, then $\bf x$ has rank at least three and we stop. 
\end{enumerate}
\end{proof}

\newcommand{\noopsort}[1]{} \newcommand{\singleletter}[1]{#1}

 \end{document}